\documentclass[sigconf]{acmart}

\usepackage{textcomp}
\usepackage{booktabs}   
\usepackage{subcaption} 
\usepackage{etoolbox}   
\usepackage{xcolor}     
\usepackage{todonotes}  
\usepackage{xspace}     
\usepackage{multirow}   
\usepackage{multicol}   
\usepackage{amsmath}
\usepackage{amsfonts}
\usepackage{amssymb}
\usepackage{amsthm}
\usepackage{algorithm}  
\usepackage[noend]{algpseudocode}  
\usepackage{tikz}       
\usepackage{pifont}     
\usepackage{enumitem}   
\usepackage{listings}
\usepackage{microtype}  
\usepackage{marvosym}   

\usepackage{float}
\newfloat{algorithm}{t}{lop}

\makeatletter
\g@addto@macro{\UrlBreaks}{\UrlOrds\do\/\do-\do.}
\makeatother

\presetkeys{todonotes}{inline}{}
\setlist{noitemsep,topsep=0pt,parsep=0pt,partopsep=0pt,leftmargin=*}

\newtoggle{printnotes}   
\newtoggle{printtodos}   
\newtoggle{printfullbisectbiggest} 
\togglefalse{printnotes} 
\togglefalse{printtodos} 
\togglefalse{printfullbisectbiggest}

\newcommand{\flit}{FLiT\xspace}
\newcommand{\mfem}{MFEM\xspace}
\newcommand{\fail}{{\color{red}\ding{56}}\xspace}
\newcommand{\pass}{{\color{green!80!black}\ding{52}}\xspace}
\newcommand{\D}{\tikz{\node[shape=circle, inner sep=.5pt]{\scriptsize\textsc{x}};}}
\newcommand{\p}{$\cdot$\xspace}
\newcommand{\Test}{\textup{\textsc{Test}}\xspace}

\newcommand{\found}{f\!ound}

\newcommand*\circled[1]{\!\!\tikz[baseline=(char.base)]{
            \node[shape=circle,draw,inner sep=.5pt] (char) {#1};}\!\!}

\newtheorem{definition}{Definition}
\newtheorem{theorem}{Theorem}
\newtheorem{assumption}{Assumption}

\newcommand{\personaltodo}[4]{{\color{#3}\todo[color=#2]{\textbf{#1:} #4}}}
\newcommand{\miketodo}[1]{\personaltodo{MIKE}{red!60!black}{white}{#1}}

\def\BibTeX{{\rm B\kern-.05em{\sc i\kern-.025em b}\kern-.08em
    T\kern-.1667em\lower.7ex\hbox{E}\kern-.125emX}}

\iftoggle{printtodos}{}{
  \presetkeys{todonotes}{disable}{}
  \renewcommand{\missingfigure}[1]{}
  
}

\newenvironment{myprintable}[1]{%
  \def\@togglevar{#1}       
  \iftoggle{\@togglevar}{   %
    \begingroup             %
  }{                        %
    \setbox 0 \vbox\bgroup  
  }                         %
}                           %
{                           %
  \iftoggle{\@togglevar}{   %
    \endgroup               %
  }{                        %
    \egroup                 %
  }                         %
}                           %

{                                   %
  \myprintable{printnotes}          %
    \par                            %
    \color{blue}                    %
}{                                  %
  \endmyprintable                   %
}                                   %

\begin{document}

\copyrightyear{2019}
\acmYear{2019}
\setcopyright{acmcopyright}

\acmConference[HPDC '19]%
{The 28th International Symposium on High-Performance Parallel and Distributed Computing}%
{June 22--29, 2019}%
{Phoenix, AZ, USA}

\acmBooktitle{The 28th International Symposium on High-Performance Parallel and Distributed Computing (HPDC '19), June 22--29, 2019, Phoenix, AZ, USA}

\acmPrice{15.00}
\acmDOI{10.1145/3307681.3325960}
\acmISBN{978-1-4503-6670-0/19/06}

\title{
  Multi-Level Analysis of Compiler-Induced Variability and Performance
  Tradeoffs
  }

\author{Michael Bentley}
\email{mbentley@cs.utah.edu}
\author{Ian Briggs}
\email{ianbriggsutah@gmail.com}
\author{Ganesh Gopalakrishnan}
\email{ganesh@cs.utah.edu}
\affiliation{\institution{University of Utah}}

\author{Dong H. Ahn}
\email{ahn1@llnl.gov}
\author{Ignacio Laguna}
\email{lagunaperalt1@llnl.gov}
\author{Gregory L. Lee}
\email{lee218@llnl.gov}
\author{Holger E. Jones}
\email{jones19@llnl.gov}
\affiliation{\institution{Lawrence Livermore National Laboratory}}

\renewcommand{\shortauthors}{Michael~Bentley,~et~al.}

\begin{abstract}

Successful HPC software applications are long-lived.
When ported across machines and their compilers,
these applications often produce different numerical results,
many of which are unacceptable.
Such variability is also a concern while
optimizing the code more aggressively to gain performance.
Efficient tools that help locate the program units (files and functions)
within which most of the variability occurs are badly needed,
both to plan for code ports and to root-cause errors due
to variability when they happen in the field.
In this work, we offer an enhanced version of the open-source
testing framework \flit to serve these roles.
Key new features of \flit include a suite of bisection algorithms
that help locate the root causes of variability.
Another added feature allows an analysis of the tradeoffs between performance
and the degree of variability.
Our new contributions also include a collection of case
studies.
Results on the \mfem finite-element library
include variability/performance
tradeoffs, and the identification of a (hitherto unknown) abnormal
level of result-variability even under mild compiler optimizations.
Results from studying the Laghos proxy application
include identifying a significantly divergent floating-point result-variability
and successful root-causing down to the problematic
function over as little as 14 program executions.
Finally, in an evaluation of 4,376 controlled injections of floating-point
perturbations on the LULESH proxy application, we showed
that the \flit framework has 100\% precision and recall in
discovering the file and function locations of the injections
all within an average of only 15 program executions.
\end{abstract}



\begin{teaserfigure}
  \centering
  \includegraphics[width=.9\textwidth]{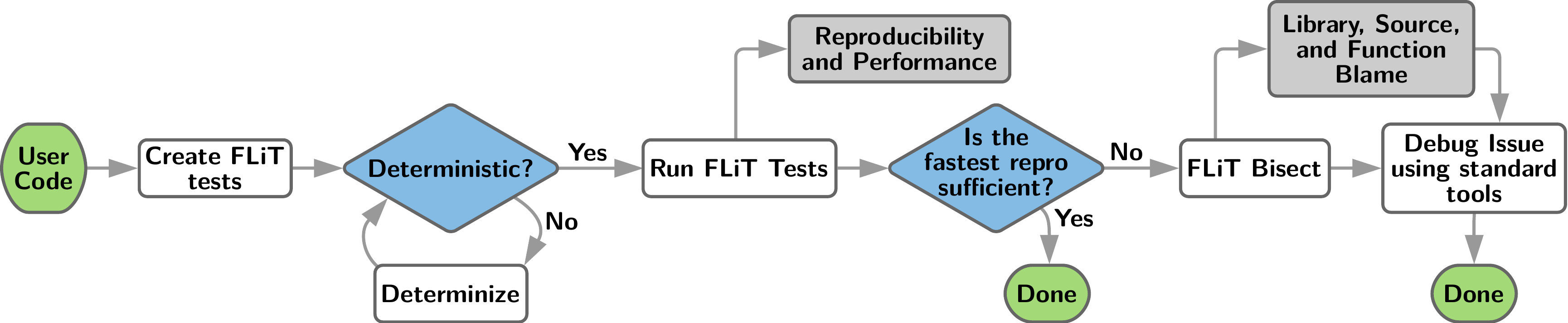}
  \Description[Multi-level workflow]{
    User code is input into creating FLiT tests.
    Then the user checks if each test is deterministic.
    If not, determinize it.
    After each test is deterministic, you run the FLiT tests which generates
    reproducibility and performance information.
    From that information, is the fastest reproducible compilation sufficiently
    fast?
    If so, then you are done.
    Otherwise, run FLiT Bisect, which isolates blame to libraries, source
    files, and functions.
    This information can then be used to debug the issue using standard tools.
    }
  \caption{
    Multi-level workflow.
    Levels are
    (1) determine variability-inducing compilations,
    (2) analyze the space of reproducibility and performance, and
    (3) debug variability by identifying files and functions causing
        variability.
    }
  \label{fig:workflow}
\end{teaserfigure}

\maketitle

\section{Introduction}                       
\label{sec:intro}
Tools and techniques that mitigate the effects of compiler-induced result-variability
are increasingly important to preserve the value of our investments in scientific
software.
As a specific example, long-lived scientific applications must be able
to take advantage of different (or newer) machines and their compilers
(as well as their optimization flags) while maintaining result integrity and
achieving higher performance.
Unfortunately, there are currently no 
techniques and tools that help designers debug field issues
that arise during such code ports,
especially in the context of large codebases and thousands of functions.
%
%
At present, designers end up wasting their time
by manually debugging field issues.
Also, code that is shipped without portability testing
may harbor the potential
to generate unacceptably
significant result variations even under
standard optimization.
An incident of this type was reported by
designers of the Community Earth System Model
(a large-scale climate simulation)~\cite{baker-work}
where the problem was noticed while porting code
to a new machine.
After weeks of painstaking investigations,
the problem turned out to be the
{\em introduction} of fused-multiply-add instructions
by the compiler, taking advantage of this new capability
offered by their target architecture.
This and other incidents reported in this paper underscore the
need for an ecosystem of freely available tools that can
help scientific programmers.
To the best of our knowledge, \flit is the first such tool.

\textbf{Definition of Reproducibility.}
Given the growing heterogeneity of hardware and software,
one cannot always define reproducibility as achieving
bitwise reproducible results.
Instead, we view a reproducible computation as one that produces a result within
an ``acceptable range'' of a trusted baseline answer.
In \flit, we rely on the application developer to provide
an acceptance testing function that (indirectly) defines this range.

\textbf{A Motivating Problem.}
Scientific HPC applications can be large and complex,
often simulating physical phenomena
for which expected outcomes are not known.
As a result, there is a particular compilation configuration 
that is trusted because it has passed the test of time 
(\textit{i.e.}, it is believed to be correct from the first version), 
and is considered the {\em baseline compilation configuration}.

When developers port applications to a different compiler or
a new version of the same compiler,
all {\em acceptably good} compilation configurations must
deliver answers empirically close to the baseline,
either based on designer experience
or in a more rigorous mathematical sense, such as
meeting an error norm.
When results deviate from acceptable levels,
support tools must help locate the issue within a short distance of the
root cause.

\textbf{Our Contributions.}
\begin{enumerate}
\item A significantly extended version of the \flit~\cite{paper:flit-iiswc} testing tool%
    \footnote{
      This version of \flit source code is available at
      \url{https://github.com/PRUNERS/FLiT.git}
      }
    that is now capable of handling real applications that
    are either sequential or contain deterministic OpenMP or
    MPI code.
\item Results that capture how performance varies versus reproducibility
      on non-trivial applications.

\item A suite of novel bisection algorithms that help identify code locations 
    responsible for result-variability,

\item A workflow
  (Figure~\ref{fig:workflow}) providing steps
  for practitioners to analyze result and performance variability.

\item Experimental validation using real-world HPC miniature applications that
      include
    Laghos~\cite{paper:laghos} and LULESH~\cite{LULESH2:changes}, and the \mfem
    library~\cite{mfem-library}.
    These studies quantitatively evaluate the effectiveness of
    our Bisect algorithms as well as empirically assess the real-world
    applicability of the workflow.

\end{enumerate}

%
%

\textbf{New \flit features.}
The new \flit features are: 
\begin{enumerate}
\item A multi-level analysis workflow supported by \flit
(Figure~\ref{fig:workflow}), resulting in 
root-cause analysis of compiler-induced
result-variability down to individual source files
and functions. Root-causing is achieved by
 FLiT's Bisect algorithms (\S\ref{sec:workflow}).

\item An assessment of the efficacy
    of Bisect on real applications and a
    fault injection study (\S\ref{sec:experimental-results});

\item The results of applying \flit, for the first time,
    on two real-world systems: MFEM and Laghos
    (\S\ref{sec:experimental-results}).
\end{enumerate}

\textbf{Compiler-Induced Variability Example.}
Compiler-induced variability is widely experienced but
seldom systematically solved.
We provide an example to help the reader better understand
the utility of a tool such as \flit.
At one stage of the development of Laghos, an open-source
simulator of compressible gas dynamics \cite{laghos-paper}, the project scientists 
were seeking higher optimizations provided by the IBM compiler, \texttt{xlc}.
Moving from optimization level \texttt{-O2} to \texttt{-O3}, the
$\ell_2$ norm of the energy over the mesh went from
129,664.9 to 144,174.9 {\em in a single iteration} --- an 11.2\% relative
difference caused merely by the optimizations.
One would expect variability around $10^{-8}\%$ or less.
Also, the density of the simulated gas became negative --- a physical
impossibility.
Even more striking was the
runtime difference: from 51.5 seconds to 21.3 seconds for the first iteration,
which is a speedup by a factor of 2.42.
In Section~\ref{sec:experimental-results}, we describe how \flit came
to the Laghos designers' rescue.

\textbf{Paper Organization and Result Highlights.}
In Section~\ref{sec:workflow}, we introduce our 
multi-level
analysis workflow and tooling
spread over three phases.
The first phase identifies which compiler
optimizations cause reproducibility problems.
The second phase helps to analyze the performance resulting
from the optimizations, thus assisting the programmer in arriving at
the most performant of acceptable solutions.
The third phase helps characterize
which functions within the code exhibit variability under compiler
optimizations, sorted by the most influential.
The last phase involves our suite of bisection algorithms.

We contribute two key assumptions that help make bisection
practical:
(1)~The \textit{Unique Error} assumption, meaning
for a particular value of variability, the
set of responsible application functions is unique.
This assumption frequently holds in practice,
as demonstrated by our results.
Without this assumption,
we will have an exponential search problem to solve.
(2)~The {\em Singleton Blame Site} assumption, which
means that a single file or function, by itself, causes variability.
In other words, it is not necessary to have two or more files
or functions to be jointly acting to induce variability.
This assumption also holds in practice, as demonstrated by our results.
The Bisect algorithm has a built-in dynamic verification assertion
that verifies this assumption.
Section~\ref{subsec:bisect-algorithm} explains how
these assumptions are central to achieving an overall $O(k \log(N))$ runtime
complexity (for $k$ ``problematic'' files/symbols) as opposed to
the $O(2^N)$ complexity, if we were to relax these assumptions.

In navigating performance and reproducibility in the \mfem library
(Section~\ref{sec:experimental-results}),
we found that 14 of 19 examples exhibited the highest
speedups with compilations that are bitwise reproducible.
Two of those 14 showed bitwise reproducibility across all tested
compilations.
These results indicate reproducibility need not always be sacrificed
for performance gains.
We demonstrate
our Bisect algorithm on {\em all} found variability-inducing compilations
from \mfem to evaluate the effectiveness of Bisect and to empirically
characterize the proclivity of a compiler to introduce variability.
For \mfem, we provide the ``best average compilation'' for each compiler
over the set of 19 \mfem examples, along with a rough idea of how often each
compiler induces variability.
Also, thanks to \flit, we have located
an unexpected result deviation in one test of \mfem which
resulted in a 180\% relative error under a mild compiler optimization.
\flit could root-cause this failure to a single function.

\flit could also discover and root-cause
a known reproducibility bug in the Laghos proxy application.
The benefit of \flit is the automated re-discovery
of this critical bug
(first located through a two week {\em ad hoc} manual search).
This automated re-discovery took only 14 application runs under Bisect,
taking only 40 minutes.

To quantify the efficacy of our Bisect algorithm even more sharply,
we implemented a custom LLVM pass to inject floating-point
perturbations in the LULESH proxy application.
We achieved precision and recall of 100\%
at identifying the source of variability, or reporting that the injection
was benign and caused no variability.
Each injection took only 15 application executions on
average during the
Bisect search to find the
function exhibiting variability.



\section{Workflow for Multi-Level Analysis}  
\label{sec:workflow}
Key to the design of \flit is a choice of
approaches and algorithms that are essential to making an impact
in today's HPC contexts.
We now present some of these choices and describe
the workflow in Figure~\ref{fig:workflow}.

\label{page:compilation-def}
We define a {\bf compilation} as a triple (Compiler, Optimization Level, Switches)
applied to a subset of source files in an application.
This triple
contains the full configuration of how to compile a source file -- as
far as optimizations and compiler options are concerned.
Our work helps hunt down compilations that cause result-variability.

\paragraph{\bf Handling vendor-specific and general-purpose compilers\/}
Vendor-provided compilers are vital to achieving high
performance, especially within newly delivered HPC machines.
Given this, \flit cannot rely on
technologies that do not generalize to many compilers and architectures.
Some such technologies are
binary instrumentation tools such as PIN
(for Intel architectures)
and instrumentation passes based on LLVM
(for LLVM-based compilers only).

\paragraph{\bf Applicability in HPC build systems\/}
Productivity-oriented approaches in HPC critically depend
on infrastructures such as Kokkos~\cite{kokkos} and RAJA~\cite{raja}
that synthesize efficient code, naturally affect loop optimizations,
and smoothly incorporate parallelism.
Framework-specific annotations burden static analysis based approaches because
each framework requires separate support and implementation.
\flit avoids this by dealing with compiled object files directly.

\paragraph{\bf Use designer-provided tests and acceptance criteria\/}
A gen\-eric tool such as \flit cannot have pre-built notions of which
results are acceptable.
Therefore \flit engineers its solutions around C++ features to require
a minimal amount of customization.
For each test, the user creates a class and defines four methods:
\begin{itemize}
  \item \textbf{\texttt{getInputsPerRun:}}
    Simply returns an integer -- The number of floating-point values taken by
    the test as input (between 0 and the maximum value of \texttt{size\_t})
  \item \textbf{\texttt{getDefaultInput:}}
    Returns a vector of test input values.
    If there are more
    values here than specified in \texttt{getInputsPerRun}, then the input
    is split up, and the test is executed multiple times, thus allowing
    data-driven testing \cite{data-driven-testing}.
  \item \textbf{\texttt{run\_impl:}}
    The actual test that takes a vector of floating-point
    values as input and returns a test result.  The test result can either be a
    single floating-point value, or a \texttt{std::string}.
    \flit provides the return type of \texttt{std::string}
    so that the user can use more complex structures
    returned, such as arbitrary meshes.
  \item \textbf{\texttt{compare:}}
    Takes in the test values from the baseline and testing compilations, and
    returns a single floating-point value.  If the two values are considered
    equal, then this function should return 0.  Otherwise, this function should
    return a positive value.  This function behaves as a metric between the two
    values and is how \flit determines if there is variability
    in a compilation compared to the baseline.

    There are two variants of this compare function, one for \texttt{long
    double} values and another for \texttt{std::string} values.
    The user is only required to
    implement the associated variant for the return type of their test.
\end{itemize}

\flit requires deterministic executions, as shown in
Figure~\ref{fig:workflow}.
On a given platform and input, we must be able
to rerun an application and obtain the same results
as measured by the user-provided \texttt{compare} function.
There are many deterministic HPC applications, even many MPI and OpenMP
applications that provide run-to-run reproducibility.
Therefore, \flit supports the use of deterministic MPI and OpenMP.
As depicted in Figure~\ref{fig:workflow},
if an application is not deterministic, then external methods can be used to make it deterministic.
For example,
one can identify and fix races with a race detector such as
Archer~\cite{archer},
or directly determinize an execution using a capture-playback framework
such as ReMPI~\cite{rempi}.

Currently, support for GPUs does not exist in \flit.
With GPUs, the scheduling of warps can
cause floating-point reassociations, thus changing execution results.
\footnote{
 There is little external control
  one can exert on GPU warp schedulers.
  }.
Given the rapid evolutions in the GPU-space, this is future work%


\subsection{Bisect Problem}
\label{subsec:bisect}


The Bisect problem handled by \flit is multifaceted:
it must help
locate  variability-inducing compilations
while also checking for acceptable execution results.
Unfortunately, modern compilers are quite
sophisticated, and their internal operation involves
many decisions such as link-time
library substitutions, the ability (or lack of) to leverage new
hardware resources, and many more such options that
affect either performance or the execution results.
This richness forces us to adopt an approach
that is as generic as possible
and consists of compiling different files at different
optimizations and drawing a final linked image from this mixture.
The granularity of mixing versions in our case is either
at a file level, or
(by using weak symbols and overriding)
at a function level%
\footnote{%
  The approach of searching by overriding symbols is one that potentially
  creates ``Frankenbinaries.''
  For example, we may link together an Intel-compiled function with a
  GCC-compiled function at differing optimization levels.
  Our symbol-based search consists of first creating various binaries (a
  one-time cost) and merely going through different linkage combinations -
  which typically takes far less time than a compilation.
}.
When we encounter a numerical result difference during our bisection search,
we allow existing tools to help with root-causing.
Thus \flit's task is to isolate the problem down to a file
or a function.

An essential practical reality is that
hundreds of functions comprise a large application spread over
multiple files.
It is possible that the compiler optimization may
have affected any subset of these functions to
cause the observed variability.
The objective of \flit's Bisect algorithm is to identify and isolate all
functions that have contributed to result-variability.

In a general sense, one faces the daunting
prospect of identifying those functions that are ``coupled,''
meaning
they must be optimized together in a certain way to
cause result-variability.
The need to identify ``coupled'' functions
would lead to a search algorithm that considers
all possible subsets of files or functions --- an exponential
problem that, if implemented as such,
would result in a very slow tool.
The singleton blame site assumption alluded to earlier reduces the search
space considerably, as discussed in more depth in
Section~\ref{subsec:bisect-analysis}.

\begin{algorithm}
  \caption{Bisect Algorithm}
  \label{alg:bisectall}
%
%
%
  \begin{tabular}{c}
  \begin{minipage}{.95\columnwidth}
  \begin{algorithmic}[1]
    \Procedure{BisectAll}{\Test, $items$}
      \State{$\found \gets \{~\}$}
      \State{$T \gets \textsc{Copy}(items)$}
      \While{$\Test(T) > 0$}
        \State{$G, next \gets \textsc{BisectOne}(\Test,~T)$}
        \State{$\found \gets \found \cup next$}
        \State{$T \gets T \setminus G$}
          \label{alg:line:remove}
      \EndWhile
      \State{\textbf{assert} $\Test(items) = \Test(\found)$}
        \label{alg:line:assert}
      \State{\Return $\found$}
    \EndProcedure
  \end{algorithmic}
  \end{minipage}
  \\
  \midrule
  \begin{minipage}{.95\columnwidth}
  \begin{algorithmic}[1]
    \Procedure{BisectOne}{\Test, $items$}
      \If{$\textsc{Size}(items) = 1$}
        \Comment{base case}
        \State{\textbf{assert} $\Test(items) > 0$}
          \label{alg:line:assert-bisectone}
        \State{\Return $items, items$}
      \EndIf
      \State{$\Delta_1, \Delta_2 \gets \textsc{SplitInHalf}(items)$}
      \If{$\Test(\Delta_1) > 0$}
        \State{\Return $\textsc{BisectOne}(\Test, \Delta_1)$}
      \Else
        \State{$G, next \gets \textsc{BisectOne}(\Test, \Delta_2)$}
        \State{\Return $G \cup \Delta_1, next$}
      \EndIf
    \EndProcedure
  \end{algorithmic}
  \end{minipage}
  \end{tabular}
\end{algorithm}

\subsection{Bisect Algorithm}
\label{subsec:bisect-algorithm}

The Bisect algorithm
(Algorithm~\ref{alg:bisectall})
follows a simple divide and conquer approach.
It takes two inputs:
(1)~$items$, which is a {\bf set} of files/functions
in the compilations to be searched over;
and
(2)~A test function \Test
that maps $items$ to a real value that is
greater than or equal to 0.
A non-zero output indicates the existence of result
variability and also
helps us sort the problematic items (files and functions)
in order of the {\em degree of variability} they induce by themselves.
It also allows us to
formulate
\iftoggle{printfullbisectbiggest}{
  Algorithm~\ref{alg:bisectbiggestk}
}{
  the \textsc{BisectBiggest} algorithm
}
(discussed in Section~\ref{subsec:bisect-biggest}).
A zero output indicates that there is no result-variability.

\begin{figure}
  \small
  \centering
  \begin{tabular}{c|cccccccccc|c}
    \toprule
    Step & \multicolumn{10}{|c|}{
      $items$ fed to \Test in Algorithm~\ref{alg:bisectall}
      }                                                    & result \\
    \midrule
    1    & 1  & 2  & 3  & 4  & 5  & 6  & 7  & 8  & 9  & 10 & \fail  \\
    2    & 1  & 2  & 3  & 4  & 5  & \p & \p & \p & \p & \p & \fail  \\
    3    & 1  & 2  & \p & \p & \p & \p & \p & \p & \p & \p & \fail  \\
    4    & 1  & \p & \p & \p & \p & \p & \p & \p & \p & \p & \pass  \\
    5    & \p & \circled{2}
                   & \p & \p & \p & \p & \p & \p & \p & \p & \fail  \\
    \midrule
    6    & \D & \D & 3  & 4  & 5  & 6  & 7  & 8  & 9  & 10 & \fail  \\
    7    & \D & \D & 3  & 4  & 5  & 6  & \p & \p & \p & \p & \pass  \\
    8    & \D & \D & \p & \p & \p & \p & 7  & 8  & \p & \p & \fail  \\
    9    & \D & \D & \p & \p & \p & \p & 7  & \p & \p & \p & \pass  \\
    10   & \D & \D & \p & \p & \p & \p & \p & \circled{8}
                                                 & \p & \p & \fail  \\
    \midrule
    11   & \D & \D & \D & \D & \D & \D & \D & \D & 9  & 10 & \fail  \\
    12   & \D & \D & \D & \D & \D & \D & \D & \D & \circled{9}
                                                      & \p & \fail  \\
    \midrule
    13   & \D & \D & \D & \D & \D & \D & \D & \D & \D & 10 & \pass  \\
    \midrule
    Result
         &    & \circled{2} 
                   &    &    &    &    &    & \circled{8}
                                                 & \circled{9} 
                                                      &    &        \\
    \bottomrule
  \end{tabular}
  \Description[Example of the BisectAll algorithm]{
    This is an example of the BisectAll algorithm on ten items.
    Out of ten elements with some causing the test function to fail, our goal
    is to find those elements causing variability.
    This figure shows how the BisectAll algorithm would find elements 2, 8, and
    9 as the elements causing the test function to fail.
    In total it takes thirteen steps to find these three elements and rule out
    all others.
    Throughout these thirteen steps, the BisectOne procedure was invoked four
    times, each time being quicker since the search space decreases each time.
  }
  \caption{
    Illustrative example of \textsc{BisectAll} (Algorithm~\ref{alg:bisectall}).
    The numbers represent tested elements.
    The dots represent elements within the current search space, but not being
    tested.
    The small x's represent elements that have been removed from the search
    space because of previous iterations of Bisect.
    The \fail means $\Test(items) > 0$ and \pass means $\Test(items) = 0$.
    The found variability-inducing items are $\{2, 8, 9\}$.
    Each row represents a separate executable
    by linking together
    the items under test from the variable compilation
    and all others from the baseline compilation.
    }
  \label{fig:bisect-example}
\end{figure}

Notice
that procedure \textsc{BisectOne}  (helper to procedure \textsc{BisectAll}) 
does not merely return the next found element.
It instead
returns a pair of two sets.
The first set contains elements that can safely be removed from future search
steps.
The second is a singleton set --- the ``found element'' in essence.
As line 2 of
\textsc{BisectOne} indicates, this means that \Test($items$)
is greater than 0, \textit{i.e.}, the presence of this
singleton set, namely $items$, in a compilation causes
result-variability.
That means we have successfully located one variability-inducing
file/function.
We now return the pair $items, items$ indicating:
(1)~that we found $items$, and (2)~we can exclude
$items$ in future searches (line 7 of
\textsc{BisectAll}).
These elements are then removed from the search space in future Bisect
searches (as seen on line~\ref{alg:line:remove}
of procedure \textsc{BisectAll} in
Algorithm~\ref{alg:bisectall}).
This removal is not necessary for the algorithm to work correctly,
or even for the complexity,
but it is merely an optimization that allows us to prune the search
space if we happen to find elements which cause the given test to pass.
This optimization is one significant deviation from
Delta debugging~\cite{paper:delta-debugging} ---
a point discussed under the
heading {\bf Assumption 2} of Section~\ref{subsec:bisect-analysis}.

As a specific example of this strategy,
notice what we do on line 9 of
\textsc{BisectOne} which is when
$\Test(\Delta_1) = 0$.
Then we suppress future testing on $G\cup \Delta_1$.

The \Test function that is passed to the Bisect algorithms is a
user-defined metric that has the following attributes:
\begin{itemize}
\item Maps a set of items to a non-negative value, $[0, \infty)$.
\item $\Test(items) = 0 \Rightarrow$ there are no variability
  causing items
\item $\Test(items) > 0 \Rightarrow$ there is at least one
  variability causing item
\end{itemize}

In Figure~\ref{fig:bisect-example}, we can see an example of running
Algorithm~\ref{alg:bisectall}.
The \pass symbol indicates an instance when $\Test(items) = 0$ and the
\fail symbol indicates $\Test(items) > 0$. 
Horizontal lines separate individual invocations of \textsc{BisectOne}.
The small X's in Figure~\ref{fig:bisect-example} refer to the extra set of elements returned by procedure \textsc{BisectOne} indicating a set of elements to discard for future search.

Although it is true that for this example, it would be cheaper to do a linear
search over the elements, a linear search would always be $O(n)$, where $n$ is
the total number of elements.
This Bisect algorithm has worst-case complexity $O(k \log n)$ and best-case
complexity $O(k \log k)$ where $k$ is the number of variability-causing
elements to find.
Section~\ref{subsec:bisect-analysis} discusses these bounds in more detail.

\subsection{Implementation of Bisect}
\label{subsec:bisect-implementation}

The Bisect search algorithm utilizes a well-known divide and conquer technique
but applying it to find the functions causing variability is nontrivial.
Note, the terms ``function'' and ``symbol'' are used interchangeably,
although symbol usually refers to a compiled version of the function.
Since the problem is to find all functions causing variability, we could group
all functions of the application and apply the Bisect
algorithm.
But, for anything more substantial than small applications, the search space
becomes too large to search effectively.
Instead, akin to how Delta Debugging~\cite{paper:delta-debugging} was extended
to work on hierarchical structures~\cite{paper:hierarchical-delta-debugging},
we perform this Bisect algorithm on a dual-level hierarchy, first by searching
for the files where the compiler caused variability, and then
searching the functions within each found file.
This hierarchical approach
allows us to reduce the search space considerably, by splitting up the
full Bisect search into much smaller separate searches.

\begin{figure}
  \centering
  \includegraphics[width=.6\columnwidth]{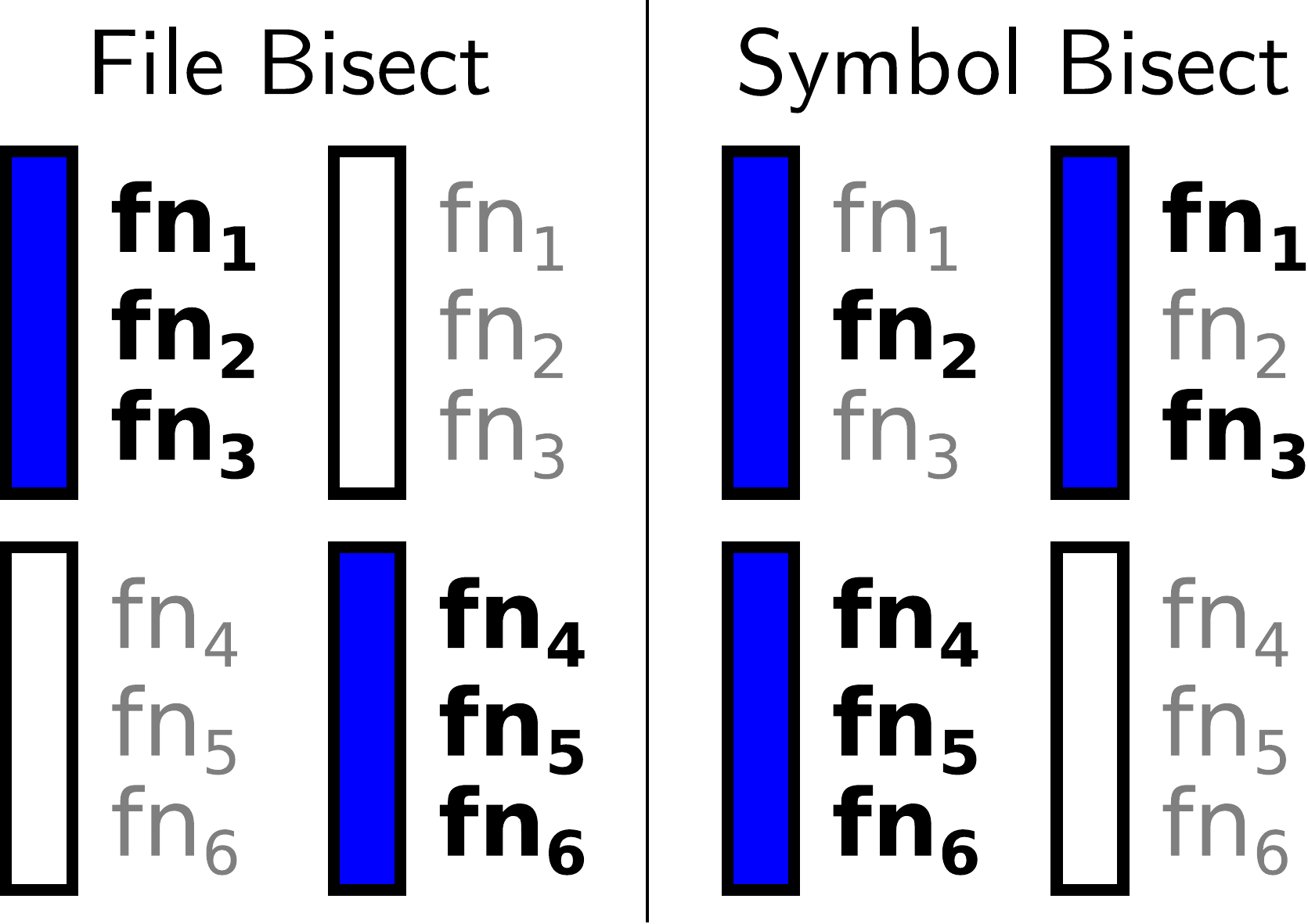}
  \Description[
    Illustration of mixing and matching object files or function symbols.
  ]{
    With the Bisect algorithm, there are two passes done.
    The first is File Bisect and then Symbol Bisect.
    This figure shows that File Bisect mixes object files, and Symbol Bisect
    mixes function symbols from within object files together.
    The Symbol Bisect is done by marking some function symbols as weak and
    linking in both compiled versions of the same object file.
    The symbols changed to weak symbols are thrown away for the strong ones
    from the other version of the same object file.
  }
  \caption{
    Highlights the difference between File Bisect and Symbol Bisect.
    File Bisect mixes compiled object files.
    Symbol Bisect marks some symbols as weak within object files and links in
    both copies of the object file.
    The functions in bold are strong symbols that are available in the
    final executable.
    Only Symbol Bisect requires the \texttt{-fPIC} flag so that we can match up
    functions arbitrarily.
    }
  \label{fig:file-bisect}
\end{figure}

The \Test function used for File Bisect
links together the object files
generated from the two different compilations, some from the
variability-inducing compilation, and the rest from the baseline compilation.
The \Test function passed into the Bisect algorithm is generated from
the baseline compilation, the variable compilation, and the full list of source
files.
When a set of source files are passed into the \Test function, those
files are compiled with the variable compilation with all others compiled with
the baseline compilation, and then the two sets of object files are linked
together.
We provide a visualization of File Bisect in the left half of
Figure~\ref{fig:file-bisect}.

It is possible that the baseline and variable compilations use different
compilers, in which case this approach depends heavily on binary compatibility
between the two compilers
\cite{intel-compatibility,gcc-compatibility}.
Since many compilers implement their own C++ standard library (since C++ 11),
one achieves binary compatibility only by forcing all compilers to
use a common implementation.
In our experiments, we chose to have all compilers use the GCC implementation
of the C++ standard library.

In the File Bisect phase,
the Bisect algorithm finds all variability-contributing 
object files when compiled with the variable compilation.
Each compiled object file comes from a single source file, and therefore can
indicate each responsible source file.

Having finished finding all variability-contributing object files,
we move on to finding the variability-inducing symbols within the found
object files
(\textit{i.e.}, methods and functions).
This second pass over symbols, called Symbol Bisect, is performed
individually on all symbols within each found variability-producing object
file.

\noindent{\bf Exploiting Linker Behavior and Objcopy:\/} The
method for selecting functions from two different versions of the same
object file is done by making use of strong and weak symbols and is shown in
the right half of Figure~\ref{fig:file-bisect}.
At link time, if there is more than one strong symbol, the linker reports a
duplicate symbol error.
If there is more than one weak symbol, then the linker is allowed to choose
which one to keep and discards the rest.
In the case there is one strong symbol and one or more weak symbols, the linker
keeps the strong symbol and discards all weak symbols.
It is the last case we utilize to select functions.
Using \texttt{objcopy}, we can duplicate an object file, and change a subset of
the strong symbols into weak symbols.
The other object file is then treated similarly, but marking the complement set
of symbols as weak.
At this point, both object files can be successfully linked together into the
executable.

However, when a compiler generates an object file,
it works under the assumption that the object file, also known as a single
translation unit, is indivisible~\cite{c-standard},
and therefore perform many optimizations based on that assumption.
%
%
%
This problem of switching the implementation of a function
has been solved in the domain of shared libraries, with the use
of \texttt{LD\_PRELOAD} and is called interposition.
%
To successfully replace all instances of one function, one must
use the \texttt{-fPIC} flag, thus disabling inlining of functions
that are callable from other translation units
(\textit{i.e.}, the globally exported symbols).
When the search reaches the Symbol Bisect phase,
the target file is recompiled with this
flag, and the result is checked.
If variability is removed by using
\texttt{-fPIC}, then the search cannot go deeper; we must
be content with reporting the file containing the variability.
We are limited, therefore, to search within the space of globally exported
symbols, since those are the only ones we can guarantee can be replaced
entirely with the desired version.

Our File Bisect and Symbol Bisect approaches are not the only ways to combine
functions from two different compilations.
For example, some compilers allow turning on and off compiler optimizations
using \texttt{\#pragma} statements.
This approach would work only for compilers with such a capability,
and would not be able to handle the situation of mixing compilations
that have two different compilers, such as GCC and the Intel compiler,
or even two different compiler versions.
Another strategy is to split the functions into separate source files.
However, this approach is non-trivial to implement and has the potential to
disable many of the optimizations that cause variability.
The final approach we considered was compiler intermediate representation,
such as LLVM IR.
This approach will work only with the compilers with which we can
perform such a pass, at the very least excluding the use of closed source
compilers such as the Intel compiler, the IBM compiler, and the PGI compiler.
For these reasons, we chose to work on combining object files after compilation
to conduct our search in File Bisect and Symbol Bisect.

We autogenerate the \Test function for Symbol Bisect using
the full set of source files,
the one source file to search,
and the full list of globally-exported symbol names from that source file.
It then marks certain symbols as weak from the two versions of the
variability-inducing object file
(compiled by \flit with \texttt{-fPIC})
and links together these two object files with the rest of the object files
compiled with the baseline compilation.

\subsection{Bisect Analysis}
\label{subsec:bisect-analysis}

Stated in a general manner,
our objective is to find all functions that contribute to the observed
variability.
The Bisect algorithm is used for both symbols and files, so here we
use a set of elements for which a \Test function can
quantify the observable variability.

Bisect is based on Delta Debugging, whose explicit goal is to
find a single \textbf{minimal set} that causes  \Test to
fail~\cite{paper:delta-debugging}.
\begin{definition} \label{def:minset}
  $Y$ is a \textbf{minimal set} of $X$,
  denoted by the boolean relation $MS(Y, X)$,
  if
  $
  \forall Z,\,
  [Y \subseteq X
  \wedge
  \Test(Y) > 0
  \wedge
  Z \subsetneq Y \Rightarrow \Test(Z) = 0]
  $.
\end{definition}
Such a minimal set is not guaranteed to be unique.
Furthermore, Delta Debugging only approximates minimal sets.
Instead of finding an arbitrary minimal set, we seek to find all elements that
contribute to the variability observed when we test all elements.
We start out by defining the elements we do not care about.
\begin{definition} \label{def:benign}
  $x$ is a \textbf{benign element} of $\,X$,
  denoted by the bool\-ean relation $B(x, X)$,
  if $\forall Y \subseteq X$, [$\Test(Y) = \Test(Y \cup \{ x \})]$.
\end{definition}
In other words, a benign element has no effect on the outcome of \Test within
the set $X$.
Using this definition of a benign element,
we define a variable elements as not benign.
\begin{definition} \label{def:all-variable}
  The set of \textbf{all variable elements} of $X$,
  denoted $AV(X)$,
  is $AV(X) \triangleq X \setminus \{ x : B(x, X) \}$
\end{definition}
This set $AV(X)$ represents the smallest set that fully explains $\Test(X)$.
Specifically, by the definition of benign elements, we see
$\Test(AV(X)) = \Test(X)$.
Finding this set $AV(X)$ is the goal of this paper and of the Bisect algorithm.
Without any assumptions or restrictions on the search space, just identifying a
single benign element $x$ requires testing against every subset of $X$ to
certify that $x$ is truly benign.
The complexity to evaluate $B(x, X)$ is $O(2^N)$ for just one element,
where $N = |X|$.
\begin{assumption} \label{assume:no-cancel}
  Errors from different sets of variable elements are distinct in magnitude.
  That is,
  $\Test(X) = \Test(Y)$ if and only if $AV(X) = AV(Y)$.
\end{assumption}
This assumption states that the only way for \Test values to match is if
the same underlying variable elements are present.
Given the nature of floating-point arithmetic, it is very unlikely for
compiler-induced variability to have the exact same magnitude.
Without this assumption, we could not do any better than brute-force
search or some approximation technique.

It is noteworthy to mention that given this assumption,
we can formulate this problem to be solved by Delta Debugging, as
follows.
Let $U$ be the universal set of all elements.
Define a new Boolean function $\Test'(Y) \triangleq [\Test(Y) = \Test(U)]$.
%
%
%
\begin{theorem}
  Let~
  $
   MS'(Y, X)
     \triangleq
  \forall Z,\,     
       [Y \subseteq X
       \wedge
       \Test'(Y)
       \wedge
       Z \subsetneq Y \Rightarrow \neg\Test'(Z)].
   $
  If Assumption~\ref{assume:no-cancel} holds, then $MS'(AV(U), U)$ and
  $\forall X,\, [X \neq AV(U) ~ \Rightarrow ~ \neg MS'(X, U)]$.
  \textup{That is}, $AV(U)$
  \textup{is the unique minimal set of} $U$.
\end{theorem}
\begin{proof}
 By the definition of $AV$, we have $\Test'(AV(U))$ is true
 because $AV(AV(U)) = AV(U)$.
 From Assumption~\ref{assume:no-cancel},
 if $Z \subsetneq AV(U)$, then $\neg \Test'(Z)$,
 since $AV(Z) \neq AV(U)$.
 Therefore $MS'(AV(U), U)$ is true.
 Now, assume $AV(U)$ is non-unique.
 Then there exists an $X \subseteq U$ such that $X \neq AV(U)$
 and $MS'(X, U)$ is true.
 This leads to a contradiction:
  \begin{description}
    \item[\textbf{Case 1:}] $AV(X) = AV(U) \subsetneq X$. \\
      But $AV(X) \subsetneq X$ and $\Test'(AV(X))$,
      therefore $\neg MS'(X, U)$.
      \Lightning

    \item[\textbf{Case 2:}] $AV(X) \subsetneq AV(U)$. \\
      But $\Test(X) \neq \Test(U)$ because $AV(X) \neq AV(U)$
      by Assumption~\ref{assume:no-cancel}.
      Therefore $\neg \Test'(X)$ and subsequently, $\neg MS'(X, U)$.
      \Lightning
  \end{description}
\end{proof}
Since Delta Debugging finds minimal sets and this minimal set is unique,
we could use Delta Debugging at this point to solve for $AV(U)$.
The complexity of the Delta Debugging algorithm is $O(k^2 \log N)$,
where $k = |AV(U)|$ and $N = |U|$.
We can do better.

\begin{assumption} \label{assume:singleton}
  \textbf{Singleton Blame Site Assumption}.
  Each variability element contributes individually.
  \begin{equation*}
    \forall x \in AV(X), \Test(\{x\}) > 0
  \end{equation*}
\end{assumption}
This assumption claims there is no situation where two or more elements
need to be tested together in order to generate a measurable variability.
In general, this is not always true.
However, we found in the domain of compiler-induced variability,
it is true in practice
--
as demonstrated by the experimental use cases in this paper.
With Assumption~\ref{assume:singleton}, we can now do Bisect search to find
each element of $AV(U)$ individually.
Each call to \textsc{BisectOne} is a logarithmic search with complexity
$O(\log N)$.
This function is called once for each element to find from $AV(U)$.
Therefore, the complexity of the Bisect algorithm is $O(k \log N)$,
again with $k = |AV(U)|$.
If $k$ is proportional to $N$ (which for this problem we have not seen to be
the case), then a linear search may outperform both Bisect search and Delta
Debugging.

\begin{table*}
 \centering
 \caption{
   Compilers used in the \mfem study with summary statistics.
   The best flags are chosen by the best average speedup across all \mfem
   examples.
   The average speedup over all 19 \mfem examples is reported and
   is calculated relative to the speed of \texttt{g++ -O2}.
  }
 \begin{tabular}{lccccc}
  \toprule
  Compiler
    & Released
      & \# Variable Runs
          & Best Flags
            & Speedup \\
  \midrule
  gcc-8.2.0
    & 26 July 2018
      & 78 of 1,288 (6.0\%)
          & \texttt{-O2 -funsafe-math-optimizations}
            & 1.097 \\
  clang-6.0.1
    & 05 July 2018
      & 24 of 1,368 (1.8\%)
          & \texttt{-O3 -funsafe-math-optimizations}
            & 1.042 \\
  icpc-18.0.3
    & 16 May 2018
      & 984 of 1,976 (49.8\%)
          & \texttt{-O2 -fp-model fast=2}
            & 1.056 \\
  \bottomrule
 \end{tabular}
 \Description[
   Table of compilers used in the \mfem study with summary statistics
 ]{
   Table of compilers used in the \mfem study with summary statistics
   The GCC compiler is version 8.2.0 released on 26 July 2018.
   Of the 1,288 \mfem compilations with GCC, 78 of them produced different
   output, which is 6.0\% of the runs.
   The flags with the largest average speedup with GCC for \mfem was
   \texttt{g++ -O2 -funsafe-math-optimizations} with an average speedup factor of
   1.097 compared with \texttt{g++ -O2}.
   The Clang compiler is version 6.0.1 released on 05 July 2018,
   Of the 1,368 \mfem compilations with Clang, 24 of them produced different
   output, which is 1.8\% of the runs.
   The flags with the largest average speedup with Clang for \mfem was
   \texttt{clang++ -O3 -funsafe-math-optimizations} with an average speedup
   factor of 1.042 compared with \texttt{g++ -O2}.
   The Intel compiler is version 18.0.3 released on 16 May 2018.
   Of the 1,976 compilations with the Intel compiler, 984 of them produced
   different output, which is 49.8\% of the runs.
   The flags with the largest average speedup with the Intel compiler for \mfem
   was \texttt{icpc -O2 -fp-model fast=2} with an average speedup factor of 1.056
   compared with \texttt{g++ -O2}.
 }
 \label{tbl:compilers}
\end{table*}

\miketodo{tbl:compilers: make sure it matches Ian's data}

\begin{figure*}[tb]
  \centering
  \begin{subfigure}[b]{\columnwidth}
    \includegraphics[width=\columnwidth]{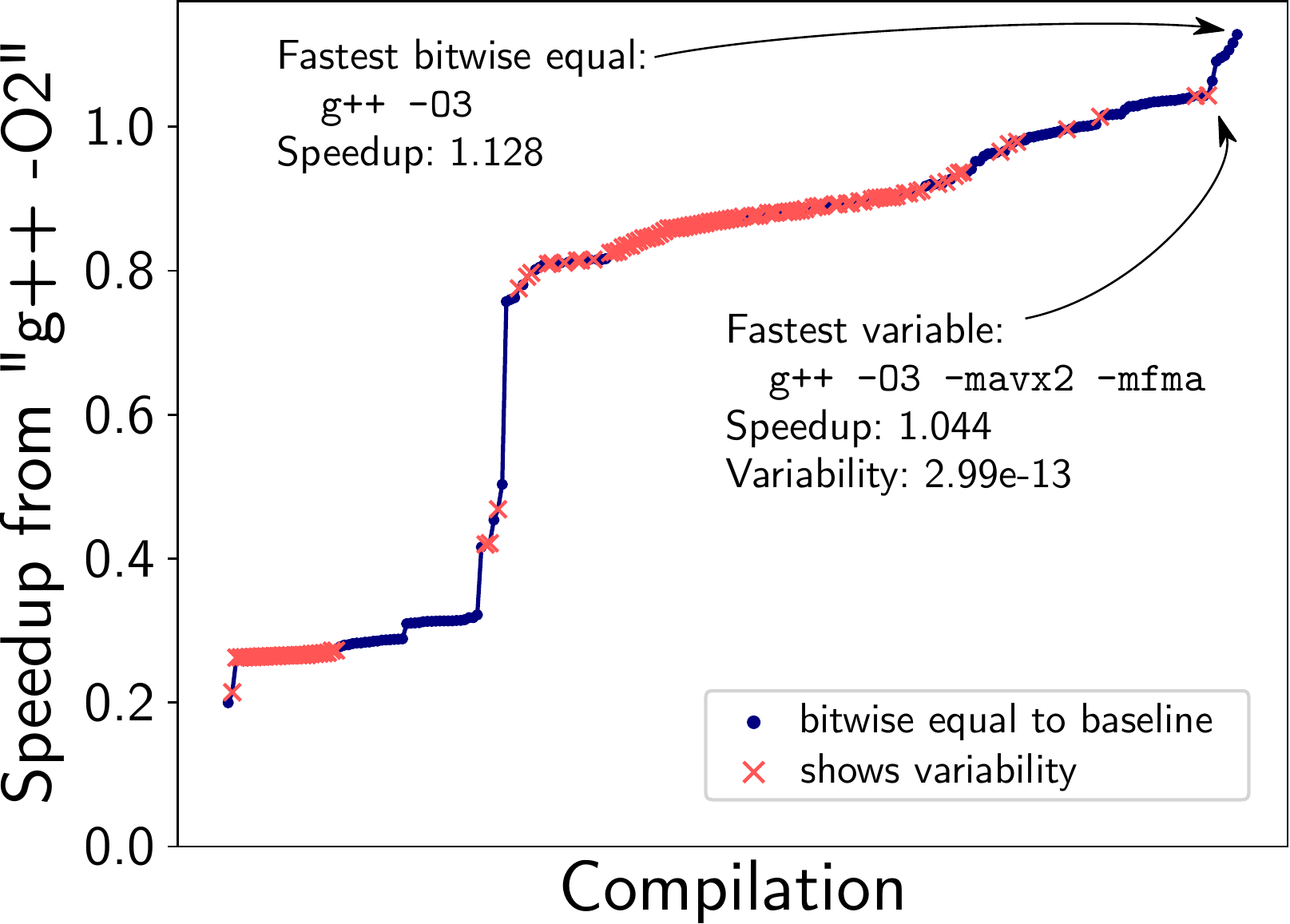}
    \caption{Example 5}
    \label{subfig:perf-repro-5}
  \end{subfigure}
  \quad
  \begin{subfigure}[b]{\columnwidth}
    \includegraphics[width=\columnwidth]{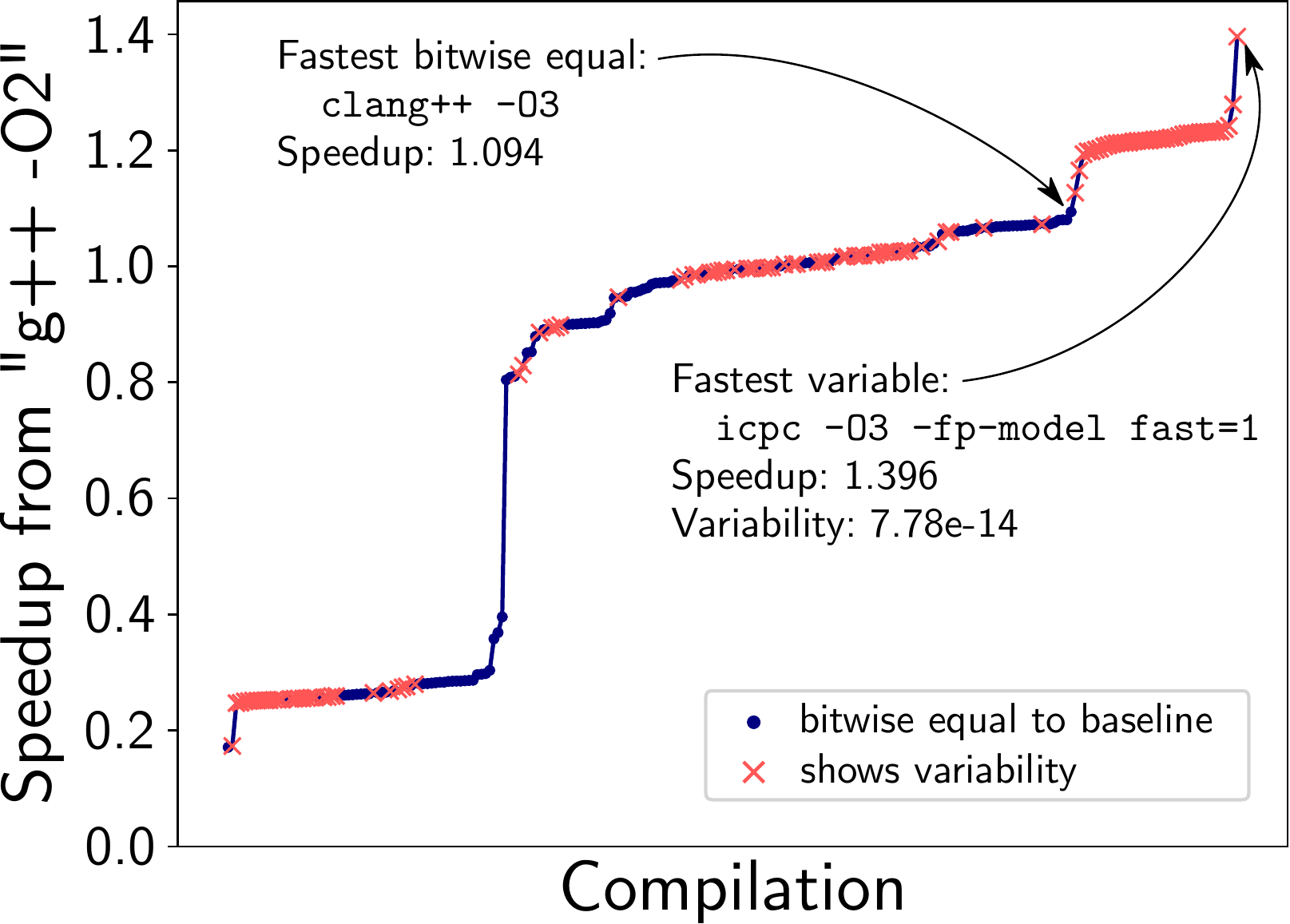}
    \caption{Example 9}
    \label{subfig:perf-repro-9}
  \end{subfigure}
  \vspace{-1em}
  \Description[
    Line plot of speedup per compilation for \mfem example 5 indicating which
    compilations produced bitwise equivalent outputs.
  ]{
    Line plot of speedup per compilation for \mfem example 5 indicating which
    compilations produced bitwise equivalent outputs.
    The compilations are sorted so that the line plot is non-decreasing.
    Compilations that produce bitwise equivalent output are represented with a
    blue circle whereas variability-producing compilations are represented with
    a red ex.
    The fastest bitwise equal compilation is \texttt{g++ -O3} with a speedup
    factor of 1.128 compared to \texttt{g++ -O2}.
    The fastest variable compilation is \texttt{g++ -O3 -mavx2 -mfma} with a
    speedup factor of of 1.044 compared to \texttt{g++ -O2} and a variability of
    $2.99 \times 10^{-13}$.
  }
  \caption{
    \mfem examples, speedup vs. compilation with compilations sorted by
    speedup.
    Both bitwise equal and variable compilations are shown.
    In~(\subref{subfig:perf-repro-5}), the fastest bitwise equal compilation was
    the fastest overall.
    In~(\subref{subfig:perf-repro-9}), the opposite is true.
    }
  \label{fig:perf-repro}
\end{figure*}

What if Assumption~\ref{assume:singleton} is not true?
We would generate false negatives.
Except, false negatives are formally checked using the assertions found in the
Bisect algorithm.
The assertion on line~\ref{alg:line:assert-bisectone} of \textsc{BisectOne}
verifies against the case when more than one element is required to cause
\Test to be positive.
It ensures that the list of found elements are each individual contributors
to variability.
The assertion on line~\ref{alg:line:assert} of \textsc{BisectAll} guarantees
that $\found = AV(items)$.
\begin{proof}
  By Assumption~\ref{assume:no-cancel}, since $\Test(\found) = \Test(items)$,
  we have $AV(\found) = AV(items)$.
  Furthermore, because of the assertion on line~\ref{alg:line:assert-bisectone}
  of \textsc{BisectOne}, we know that each element of $\found$ is a variable
  element.
  Therefore, $\found = AV(items)$.
\end{proof}

Despite this simple proof, the result is profound.
If Assumption~\ref{assume:no-cancel} holds, and the assertions in the
Bisect algorithm pass, then there are no false negatives,
meaning we have found all variability elements.
And this dynamic verification requires $2 + k$ extra calls to $\Test$
(though really $1 + k$ calls because $\Test(items)$ can be memoized).
However, if the assertion fails, then either
Assumption~\ref{assume:no-cancel}
or
Assumption~\ref{assume:singleton} are false,
in which case the user is notified that there may be false negative results.
Also worth noting is that because of the assertion on
line~\ref{alg:line:assert-bisectone},
we guarantee that $\found$ are all variable elements,
meaning it is impossible to get false positive results.

\subsection{The Bisect Biggest Algorithm}
\label{subsec:bisect-biggest}

\iftoggle{printfullbisectbiggest}{
  Before running the Bisect algorithm, we cannot know a priori the value of $k$
  -- the number of functions that exhibit variability.
  The worry is that if $n$ is large, then perhaps $k$ is large too.
  For this reason, we also developed a \textsc{BisectBiggest} algorithm where
  the value of $k$ can be fixed to a predetermined value, and we can exit the
  algorithm early.

  \def\allFiles{all_{\!f\!iles}}
\def\allSymbols{all_{sym}}
\def\foundFiles{f\!ound_{\!f}}
\def\foundSymbols{f\!ound_{s}}
\def\frontierFile{f\!rontier_{\!f\!ile}}
\def\frontierSymbol{f\!rontier_{\!sym}}
\def\kthScore{kthScore}
\def\scoreFile{score_{\!f}}
\def\scoreSymbol{score_{s}}
\def\currentFiles{f\!iles}
\def\currentSymbols{syms}

\def\PriorityQueue{\textsc{PriorityQueue}}
\def\Push{\textsc{Push}}
\def\Pop{\textsc{PopMax}}
\def\Size{\textsc{Size}}
\def\AllSymbols{\textsc{AllSymbols}}
\def\IsNotEmpty{\textsc{IsNotEmpty}}
\def\SortByScore{\textsc{SortByScore}}
\def\SplitInHalf{\textsc{SplitInHalf}}

\begin{algorithm}
  \caption{Bisect the biggest $k$ contributing functions}
  \label{alg:bisectbiggestk}
  \begin{algorithmic}[1]
    \Procedure{BisectBiggestK}{\Test, $\allFiles$, $k$}
      \State{$\foundFiles \gets \{~\}$}
      \State{$\foundSymbols \gets \{~\}$}
      \State{$\frontierFile \gets \PriorityQueue()$}
      \State{$\Push
        \left(
          \frontierFile,
          \left(
            \Test(\allFiles),
            \allFiles
          \right)
        \right)$}
      \State{$\kthScore \gets 0$}
      \While{$\IsNotEmpty(\frontierFile)$}
        \State{$\scoreFile, \currentFiles \gets \Pop(\frontierFile)$}
        \If{$\scoreFile \leq \kthScore$}
          {\textbf{break}}
        \EndIf
        \If{$\Size(\currentFiles) = 1$}
          \State{$\foundFiles \gets \foundFiles \cup \currentFiles$}
          \State{$\frontierSymbol \gets \PriorityQueue()$}
          \State{$\allSymbols \gets \AllSymbols(\currentFiles)$}
          \State{$\Push
            \left(
              \frontierSymbol,
              \left(
                \Test(\allSymbols),
                \allSymbols
              \right)
            \right)$}
          \While{$\IsNotEmpty(\frontierSymbol)$} \label{alg:line:immediate-recursion}
            \State{$\scoreSymbol, \currentSymbols \gets \Pop(\frontierSymbol)$}
            \If{$\scoreSymbol \leq \kthScore$}
              {\textbf{break}}
            \EndIf
            \If{$\Size(\currentSymbols) = 1$}
              \State{$\foundSymbols \gets \foundSymbols \cup \currentSymbols$}
              \State{$\SortByScore(\foundSymbols)$}
              \If{$\Size(\foundSymbols) \geq k$}
                \State{$\kthScore \gets \Test(\foundSymbols[k])$}
              \EndIf
            \Else
              \State{$\Delta_1, \Delta_2 \gets \SplitInHalf(\currentSymbols)$}
              \State{$\Push(\frontierSymbol, (\Test(\Delta_1), \Delta_1))$}
              \State{$\Push(\frontierSymbol, (\Test(\Delta_2), \Delta_2))$}
            \EndIf
          \EndWhile
        \Else
          \State{$\Delta_1, \Delta_2 \gets \SplitInHalf(\currentFiles)$}
          \State{$\Push(\frontierFile, (\Test(\Delta_1), \Delta_1))$}
          \State{$\Push(\frontierFile, (\Test(\Delta_2), \Delta_2))$}
        \EndIf
      \EndWhile
      \State{\Return $(\foundFiles, \foundSymbols)$}
    \EndProcedure
  \end{algorithmic}
\end{algorithm}

  The \textsc{BisectBiggest} algorithm can be seen in Algorithm~\ref{alg:bisectbiggestk}.
  At its core is Uniform Cost Search (UCS) \cite{uniform-cost-search}.
  But, instead of finding all files, then searching each file afterwards,
  the symbols are searched immediately upon finding the largest contributing file
  on line~\ref{alg:line:immediate-recursion}.
  The key point here is that if the \Test value for the current set of files or
  symbols is less than or equal to the \Test value of the $k^{\text{th}}$ most
  variability inducing function found so far, then we can safely exit.
  The ability to exit early stems from the following assumption:

  \begin{assumption} \label{assume:strict-no-cancel}
    If function $f$ is in the source file $F$,
    then $\Test(\{f\}) \leq \Test(\{F\})$.
  \end{assumption}
  That is to say that there is no function within a file with a higher
  variability score than the file itself.
  This is not always true since errors can partially cancel each other out.
  But without this assumption, we would need to find all variability-inducing
  files regardless of the value chosen for $k$.
  Unfortunately, we are not able to perform the same assertion in
  Algorithm~\ref{alg:bisectbiggestk} that is done in
  Algorithm~\ref{alg:bisectall} because the list of found elements is incomplete
  by design.

  Despite not having the dynamic verification assertion, this approach has the
  possibility to improve runtime significantly if it is only desired to grab the
  first one or two most variability inducing functions.
  It should be noted, however, that if the chosen value for $k$ is close to the
  true value for $k$, then using Algorithm~\ref{alg:bisectbiggestk} is more
  expensive than Algorithm~\ref{alg:bisectall}.
}{
  Along with the Bisect algorithm that finds all variability-inducing files and
  functions, we developed an algorithm that can search for the biggest $k$
  contributors where the user can choose the value for $k$.
  This variant is based on Uniform Cost Search and can exit early.
  Upon finding the largest contributing file, it immediately recurses to find
  the $k$ largest contributing symbols.
  When a file or symbol is found to have a smaller \Test value than the
  $k^{\text{th}}$ found symbol's \Test value, it exits early.
  It is not able to dynamically verify assumptions, but can significantly improve
  performance if only the top few most contributing functions are desired, and
  there happen to be many more than that to find.
}

\begin{figure*}[tb]
  \centering
  \includegraphics[width=\textwidth]{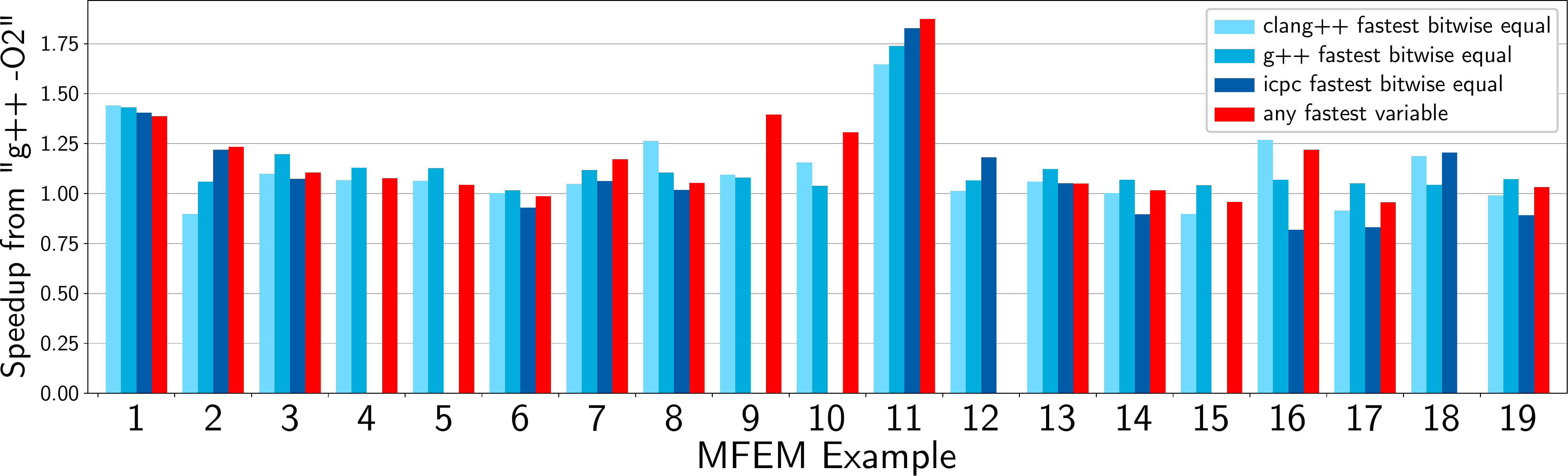}
  \Description[
    \mfem test performance histogram of the fastest compiled executables.
  ]{
    \mfem test performance histogram of the fastest compiled executables.
    There are four categories, of which the largest speedup from each category
    is displayed.
    The first three categories are the bitwise reproducible Clang, GCC, and the
    Intel compilations, and the final category is executables exhibiting
    result-variability.
    Tests 4, 5, 9, 10, and 15 have no results for the reproducible Intel
    compiler column since all Intel compilations exhibited result-variability.
    Tests 12 and 18 have no variability results, as all compilations resulted
    in bitwise equivalent outputs.
    Most of the time, bitwise reproducible compilations outperform
    variability-producing compilations.
    Tests 9 and 10 were the only two tests where variability-producing
    compilations gave a significant performance boost.
  }
  \caption{
    Performance histogram of the fastest compiled executable from each
    category for each \mfem test.
    The left three blue bars for each example represent the most
    performant bitwise equal execution, with the right red bar being the
    most performant execution exhibiting variability (combined from the three
    compilers).
    Missing bars mean there were no results in that category.
    Examples 12 and 18 had no compilations that produced variability.
    Examples 4, 5, 9, 10, and 15 are missing the Intel compiler bar, because
    variability was introduced by the Intel link step, regardless of
    optimization level or switches.
    }
  \label{fig:perf-histogram}
\end{figure*}

\section{Experimental Results}               
\label{sec:experimental-results}
We performed three evaluations of \flit: \mfem, Laghos, and LU\-LESH.
We applied \flit to \mfem to view the speed and variability space;
then we applied \flit Bisect on all found variant compilations.
The second evaluation is a real-world case study running \flit Bisect on the
Laghos codebase with an unknown issue with variability.
\miketodo{
  unknown issue?
  It was unknown when we started working on it, but was found manually.  I
  think it would be best to say it was an unknown issue that was confirmed
  using a manual approach of debugging the issue.
}
Finally, we used an LLVM pass to modify floating-point operations in the
compilation of the LULESH miniapp to evaluate precision and recall of
the Bisect algorithm.

%
%

\begin{figure}[tb]
  \centering
  \includegraphics[width=\columnwidth]{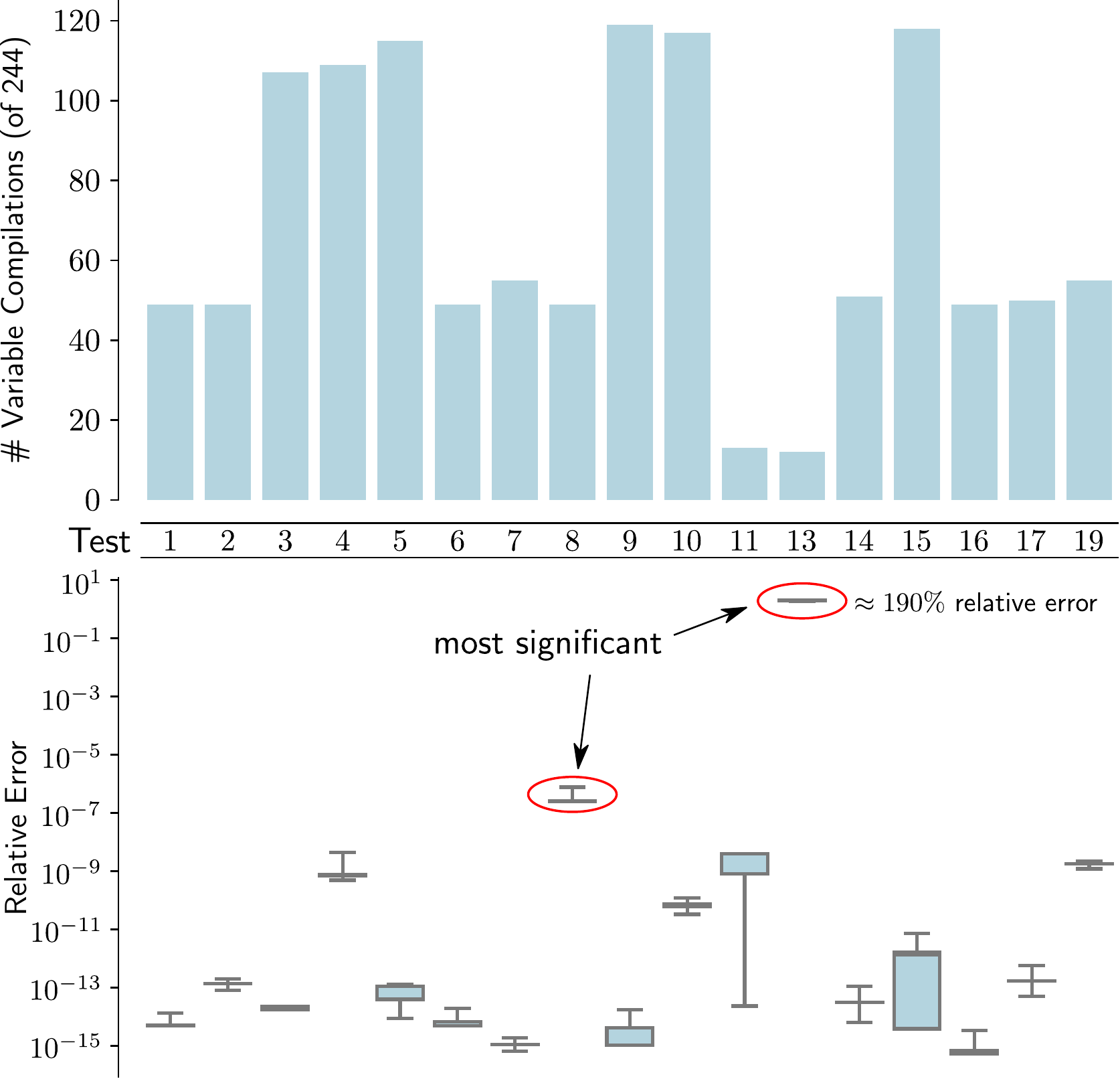}
  \Description[MFEM boxplot of found variability]{
    MFEM boxplot of found variability.
    A barchart on top shows how many variability-inducing compilations out of
    244 were found by \flit for each test.
    A boxplot on bottom shows, on a vertical logscale, the range of relative
    $\ell_2$ errors induced by the different compilations.
    Tests 12 and 18 are omitted because they had no found variabilities.
    Of note, tests 8 and 13 had significant variability, with test 8 having a
    relative error around $10^{-6}$ and test 13 having approximately 190\%
    relative error.
    }
  \caption{
    MFEM found variability.
    For each test, the top bar chart shows the number of variability-inducing
    compilations out of 244 found by \flit.
    The bottom boxplot has a vertical logscale and shows the range of relative
    $\ell_2$ errors induced by the different compilations.
    Tests 12 and 18 are omitted because they had no found variabilities.
    }
  \label{fig:mfem-variability}
\end{figure}

\subsection{Performance vs. Reproducibility Case Study}
\label{subsec:performance-vs-repro}

\mfem is a finite element library poised for use in high-performance
applications.
\flit was used with three mainstream compilers to view the
tradeoff between reproducibility and speed, as seen in
Figure~\ref{fig:perf-repro}.
In Figure~\ref{fig:perf-histogram} we examine the fastest non-variant
compilations given by each compiler with the fastest variant overall.

The \mfem library comes with 19 end to end examples of how to use the
framework, which is what we used as test cases in \flit.
These examples include the use of MPI, which \flit now supports.
Each example produces calculated values over a full mesh or volume.
The comparison function used the $\ell_2$ norm of the mesh difference,
$||baseline - actual||_2$.

Using \flit, we compiled \mfem using the {\tt g++},
{\tt clang++}, and {\tt icpc} compilers as listed in Table~\ref{tbl:compilers}.
For these compilers, we paired a base optimization level,
{\tt -O0} through {\tt -O3},
with a single flag combination, taken from the list used
in~\cite{paper:flit-iiswc}.
This cartesian product leads to 244 compilations,
and with 19 test cases results in
a total of 4,636 experimental results.
Looking at a single MFEM example and ordering the compilations from slowest to
fastest, we get graphs similar to those found in Figure~\ref{fig:perf-repro}.
The points marked with a blue circle compare equal to the
baseline results from {\tt g++ -O0}, and those with a red X exhibit variability.
For \mfem example 5 (Figure~\ref{subfig:perf-repro-5}),
the fastest compilation with bitwise equal results
was 12.8\% faster than \texttt{g++ -O2}.
This example was not an outlier;
we find similar results in 14 of the 19 examples
(Figure~\ref{fig:perf-histogram}).
This finding contrasts with Figure~\ref{subfig:perf-repro-9}, which has the
variant compilations grouped near the top and showing a significant speedup
over the fastest functionally equivalent compilation.

While these plots give detail to individual experiments,
Figure~\ref{fig:perf-histogram} shows a bigger picture.
Each grouping shows the fastest non-variant compilation and the fastest variant
compilation in regards to a single experiment.
Once again, 14 out of 19 experiments show non-variant compilations to 
be also the fastest.
Variant compilations are noticeably faster than non-variants in
only 2 of the groupings.

The magnitude of the observed result-variability is also important to consider.
In Figure~\ref{fig:mfem-variability}, we see the min, median, and max of the
relativized errors observed by the different compilations of each \mfem
example.
The errors were normalized by dividing by the $\ell_2$ norm of the baseline
mesh values.
Examples 8 and 13 showed significant variability, and are examined further
using Bisect.

\subsection{Bisect}
\label{subsec:bisect-results}

\flit found 1,086 compilations which lead to variant results,
each of which were explored by \flit Bisect.
These searches were over a non-trivial codebase.
An overview of the success rate of Bisect is available in
Table~\ref{tbl:mfem-compiler-bisect}.


\begin{table}
  \caption{
    Compiler characterization of Bisect with \mfem.
    Only those runs that succeeded with File Bisect went on to perform Symbol
    Bisect.
    A failure here means the resulting mixed executable crashed.
    }
  \centering
  \small
  \begin{tabular}{lcccc}
    \toprule
      & g++   & \!clang++\! & icpc    & total    \\
    \midrule
    average test executions\!
      & 64    &   29    & 27      & 30       \\
    File Bisect successes\!
      & 78/78 & 24/24   & 778/984 & 880/1,086 \\
    Symbol Bisect successes\!
      & 51/78 & 24/24   & 585/778 & 660/880  \\
    \bottomrule
  \end{tabular}
  \Description[\mfem Bisect measurements by compiler]{
    Table of \mfem Bisect measurements by compiler.
    The average number of test executions for GCC, Clang, and the Intel
    compiler were 64, 29, and 27 respectively with a total average of 30.
    The number of File Bisect successes for GCC, Clang, and the Intel compiler
    were 78 of 78, 24 of 24, and 778 of 984 respectively with a total of 880 of
    1,086.
    Only those that passed File Bisect went on to Symbol Bisect
  }
  \label{tbl:mfem-compiler-bisect}
\end{table}

The \mfem library contains almost 3,000 functions which are exported symbols,
as seen in Table~\ref{tbl:mfem-bisect-results}.
The \flit Bisect approach depends only on the number of source files and
functions, as opposed to static and dynamic analysis approaches that rely on the depth and breadth of the call tree.
While this size of 3,000 functions is daunting for a linear search, the
Bisect approach used an average of 30 executions including the verification
assertion.
\flit was able to isolate the variability to the file level 80\% of the time,
and of those was able to isolate the variability to the symbol level 75\% of
the time.

Two findings were significant enough to be reported back to the \mfem team and
are currently under further investigation.

\paragraph{Finding 1:}
\mfem example 8 is an iterative algorithm with a stopping criteria of
$10^{-12}$, yet converges to a value that has an absolute error of $10^{-6}$,
meaning it converged differently because of compiler optimizations.
\flit Bisect found all nine functions causing the variability for example 8, each
performing matrix and vector operations.
The compilations were
\texttt{icpc -O2},
\texttt{icpc -O3},
\texttt{g++ -O2 -mavx2 -mfma},
\texttt{g++ -O3 -mavx2 -mfma}, and
\texttt{g++ -O3 -funsafe-math-optimizations}.
FMA is a likely culprit as well as vectorization.

\paragraph{Finding 2:}
Example 13 had the most substantial variability by far, having between 183\% to
197\% relative error.
\flit Bisect found only one function to contribute to variability, a function
that calculates $M = M + a A A^\top$ with $a$ being a scalar, and $M$ and $A$
being dense square matrices.
This function is implemented in a straightforward manner using nested for
loops.
The compilations responsible enable AVX2, FMA, and higher precision
intermediate floating-point values.
Therefore, we suspect FMA, vectorization, and higher precision intermediates to
be the reasons for the variability.

\subsection{Characterization of Compilers}
\label{subsec:compiler-characteristics}

From this two-part experiment, we can assess the compilers predilection for
speed, variability, and compatibility.

The maximum available speedup for a single example ranges from a factor of 1.02 
to 1.87 relative to the {\tt g++ -O2} compilation.
But each example has its own best compilation.
Since \mfem is a library, it is better to see which compilation lead to the best
average speedup across all examples to cover all use cases.
The best average compilations, separated by compiler,
can be seen in Table~\ref{tbl:compilers}, in which {\tt g++} comes in first
with a speedup factor of 1.097.
Note, all three of these fastest average compilations have variability induced
on at least one example.

In that same Table is the percentage of compilations which caused variability.
The most invariant compiler is {\tt clang++} with only 1.8\% of compilations
deviating from the baseline.
The most variant compiler, producing almost half variable compilations (at 49.8\%),
is the Intel compiler, {\tt icpc}.
Intel's compiler went from a distant second in speed to last in variability.

By examining the Bisect results more closely, we discovered some issues that
drove the 20\% failure rate of File Bisect.
When {\tt icpc} and {\tt g++} object files were linked together, the resulting
executable would sometimes fail with a segmentation fault.
While Intel claims compatibility with the
GNU compiler~\cite{intel-compatibility},
this does not seem to always hold.

\begin{table}
  \caption{
    General statistics of code used by the \mfem examples.
    }
  \centering
  \begin{tabular}{lr}
    \toprule
    source files               &      97   \\
    average functions per file &      31   \\
    total functions            &   2,998   \\
    source lines of code       & 103,205   \\
    \bottomrule
  \end{tabular}
  \Description[General statistics of code used by the \mfem examples]{
    General statistics of code used by the \mfem examples.
    \mfem has 97 source files, 31 functions per file on average, 2,998 total
    functions, and 103,205 lines of code.
  }
  \label{tbl:mfem-bisect-results}
\end{table}

\subsection{Penetration into Laghos}
\label{subsec:laghos}

The issue found by the developers of Laghos manifested when they compiled with
IBM's {\tt xlc++} compiler at {\tt -O3}.
Given the code, Bisect was able to find an issue
not related to floating-point
that was already fixed in another branch.
After fixing that problem, we were able to isolate the problem down to the
function level.
%

The tool developers trusted the results from both {\tt g++ -O2} and
{\tt xlc++ -O2} when using their branch of the code.
We used a public branch of the code in an attempt to reproduce the results they
had.
In our runs, all results were the special floating point value $NaN$.
Using Bisect, we narrowed this down to the two visible symbols closest to the
issue.
The source code in question was \verb|#define xsw(a,b) a^=b^=a^=b|, which
evokes undefined behavior in C++.
Bisect identified these two function in 45 program executions.
The developers confirmed the bug, which they had fixed in their
version.
While this may appear to be a case of finding a bug yet again, the fact that our
automated Bisection-based search found this issue must be viewed as a step forward,
considering that the manual process by which the developers initially found this issue is
``hit or miss'' and requires expert's time to be spent.

After fixing this issue, we achieved results agreeing with the developer-stated results
for both the trusted compilation and the variant {\tt xlc++ -O3} compilation.
We ran many variants of Bisect to evaluate the speed and effectiveness of
\textsc{BisectAll} and \textsc{BisectBiggest},
as can be seen in Table~\ref{tbl:laghos-2nd}.
By limiting either the digit sensitivity of our \texttt{compare} function
or the $k$ value of \textsc{BisectBiggest}
($k = \text{all}$ refers to using the traditional Bisect algorithm),
the number of runs varied from 69 to 14,
all of which were able to identify the most significant variability-inducing
function.
In the function pointed to was an exact comparison to 0.0 in an if statement.
The value compared against 0.0 had small variability,
but the difference in branching resulted in significant application variability.
Changing this to an epsilon based comparison gave results close to the
trusted results, even under {\tt xlc++ -O3}.

\begin{table}
  \caption{
    Bisect statistics of the Laghos experiment.
    The baseline compilation is provided, with the compilation under test being
    \texttt{xlc++ -O3} versus the result of \flit Bisect.
    The \textit{strict} qualifier refers to the additional flag
    \texttt{-qstrict=vectorprecision}.
    We restrict the comparison to compare only the number of digits in the
    digits column.
    The $k$ value is
    how many of the most contributing functions Bisect is asked to find.
    }
  \centering
  \small
  \begin{tabular}{p{1.5cm}c|ccc|ccc|ccc}
    \toprule
    baseline & digits
      & \multicolumn{3}{c|}{\# files}
               & \multicolumn{3}{c|}{\# funcs}
                    & \multicolumn{3}{c}{\# runs}
    \\
      & \multicolumn{1}{r|}{$k$:}
        & 1 & 2 & \!all\!
        & 1 & 2 & \!all\!
        & 1 & 2 & \!all\!
    \\
    \midrule
    \multirow{4}{*}{\texttt{g++ -O2}}
      &  2  & 1 & 1 & 1 & 1 & 1 & 1 & 18 & 18 & 14 \\
      &  3  & 1 & 1 & 1 & 1 & 1 & 1 & 18 & 18 & 14 \\
      &  5  & 1 & 1 & 1 & 1 & 1 & 1 & 18 & 18 & 14 \\
      & all & 2 & 3 & 5 & 1 & 2 & 7 & 28 & 37 & 57 \\
    \midrule
    \multirow{4}{*}{\texttt{xlc++ -O2}}
      &  2  & 1 & 1 & 1 & 1 & 1 & 1 & 18 & 18 & 14 \\
      &  3  & 1 & 1 & 1 & 1 & 1 & 1 & 18 & 18 & 14 \\
      &  5  & 1 & 1 & 1 & 1 & 1 & 1 & 18 & 18 & 14 \\
      & all & 2 & 3 & 6 & 1 & 3 & 7 & 28 & 37 & 69 \\
    \midrule
    \multirow{4}{*}{%
      \parbox{2cm}{%
        \texttt{%
          xlc++ -O3 \hfill \\
          }
          strict
        }
      }
      &  2  & 1 & 1 & 1 & 1 & 1 & 1 & 18 & 18 & 14 \\
      &  3  & 1 & 1 & 1 & 1 & 1 & 1 & 18 & 18 & 14 \\
      &  5  & 1 & 1 & 1 & 1 & 1 & 1 & 18 & 18 & 14 \\
      & all & 2 & 3 & 5 & 1 & 2 & 5 & 28 & 39 & 60 \\
    \bottomrule
  \end{tabular}
  \Description[Table of Bisect statistics of the Laghos experiment]{
    Table of Bisect statistics of the Laghos experiment.
    The number of digits compared in the comparison function was 2, 3, 5, and
    all digits.
    All digits is basically checking for identical results.
    The value of $k$ for the \textsc{BisectBiggest} $k$ algorithm was chosen as
    1, 2, and all, with all meaning falling back to the \textsc{BisectAll} algorithm.
    Three different baselines were tested, \texttt{g++ -O2}, \texttt{xlc++
    -O2}, and \texttt{xlc++ -O3 -qstrict=vectorprecision}.
    For all three baselines, and digits less than all, only one file and one
    function was found, with a total of 18 runs for $k$ being 1 or 2 and 14
    runs with using \textsc{BisectAll}.
    When all digits were considered, up to 7 functions were found with many
    more executions, going for 28, 37, and 57 respectively for $k$ being 1, 2,
    and all for the first baseline compilation, \texttt{g++ -O2}.
    The other two baselines had even more executions necessary to find the
    variability-inducing files and functions.
  }
  \label{tbl:laghos-2nd}
\end{table}

\subsection{Injection Study}
\label{subsec:injection}

We performed controlled injections of floating-point variability at all
floating-point code locations to quantify the accuracy of our tool.
%

Our injection framework is based on the LLVM compiler~\cite{paper:llvm-compiler}
and introduces an additional
floating-point operation in a given floating-point instruction of the LLVM
intermediate representation (IR).
More formally, given a target floating-point instruction of the form
$x \textit{ OP } y$, where $x$ and $y$ are floating-point operands, and
\textit{OP} is a basic floating-point operation (+,-,*,/), we introduce an
additional operation $x \textit{ OP' } \epsilon$, where \textit{OP'} is also a
basic floating-point operation and $\epsilon$ is chosen from a uniform
distribution between 0 and 1.
For example, assuming that the target instruction is
\begin{equation*}
z = x * y,
\end{equation*}
after the injection, the resulting operation is:
\begin{equation*}
z = (x + \text{1e-100}) * y.
\end{equation*}
In this example, \textit{OP'} is the addition operation and $\epsilon$ is 1e-100.

Our variability injection framework requires two passes.
The first pass identifies potential \textit{valid injection locations}; an
injection location is defined by a file, function and floating-point instruction
tuple in the program.
The second pass injects in a user-specified location, using
a specific $\epsilon$ and operation \textit{OP'}.
We perform the injections at an early stage during the LLVM optimization step.
Our goal is to introduce variability before optimizations take place.

For our evaluation, we used the benchmark called
Livermore Unstructured Lagrangian Explicit Shock Hydrodynamics (LULESH).
This LULESH benchmark contains 5,459 source lines of code, in which there are
1,094 floating point operations.
For each of these operations, we did four injection runs, one for each possible
\textit{OP'}.

Under our evaluation criteria as seen in Table~\ref{tbl:injection},
we deem a symbol reported by \flit Bisect to be
exact of the source function where the injection occurred;
this occurred 2,690 times.
We also count indirect finds, which is when the source function is not a visible
symbol but Bisect was able to find the visible symbol which used the injected
function, what happened 984 times.
This indirection can occur for several reasons, with the majority coming from functions
which were inlined or otherwise not exported as a strong symbol.
We also count wrong finds and missed finds, which are false positives and false
negatives.
Both of these categories yielded no results in our runs.
The final category is when the injection was not measurable.
A non-measurable result is when the injection did not change the output of
LULESH, which account for 702 of the runs.
A non-measurable result can occur when the injection was in code that was not
run,
or if the injected variation did not affect the final result.

\begin{table}
  \caption{
    Success statistics of the LULESH compiler perturbation injection
    experiment.
    Indirect finds are when the injected function is not in the search space
    but we successfully report the closest global function that calls it.
    Wrong finds are when the reported function does not induce variability.
    Missed finds are when variability occurs, but we do not report the
    functions responsible.
    Not measurable indicates a benign injection.
    }
  \centering
  \begin{tabular}{lr}
    \toprule
    Category         & Count \\
    \midrule
    exact finds      & 2,690 \\
    indirect finds   &   984 \\
    wrong finds      &     0 \\
    missed finds     &     0 \\
    not measurable   &   702 \\
    \midrule
    total            & 4,376 \\
    \bottomrule
  \end{tabular}
  \Description[
    Success statistics of the LULESH compiler perturbation injection experiment
  ]{
    Table of success statistics of the LULESH compiler perturbation injection
    experiment.
    There were 4,376 total experiment runs.
    Of them, there were
    2,690 exact function finds,
    984 indirect finds,
    zero wrong finds,
    0 missed finds,
    and
    702 injections that showed no variability on the output.
  }
  \label{tbl:injection}
\end{table}

\subsection{MPI Support}

All experiments described in this paper were run sequentially.
However, in Figure~\ref{fig:workflow}, we specify runtime determinism as the
only prerequisite, meaning that
\flit can be extended to run on a
deterministic platform.

Currently, \flit supports deterministic MPI.
To test this path,
we repeated a randomized sampling of the \mfem experiment
with MPI running under 24 processes.

The first step was to give us high confidence that \mfem under MPI is
deterministic for the 19 provided examples.
This evaluation was done by performing 100 executions of each test and checking
the full matrix output for bitwise equivalence.
Unfortunately, only 17 of the 19 tests were able to be easily wrapped so that
the \flit framework could call \texttt{MPI\_Init} and \texttt{MPI\_Finalize}
(tests 17 and 18 could not be
accommodated).
All 17 converted parallel tests passed this verification, so we
have high
confidence that \flit would work well with \mfem under MPI.

Next, we wanted to determine the effects of adding parallelization to the \mfem
examples.
That is to say, how do the results from the parallel execution compare against
the sequential run?
We found that in the 17 used tests, increasing the parallelism changed the
result, as measured by the $\ell_2$ norm of the result.
We believe this is due to increasing or decreasing the grid density when
performing domain decomposition.
Regardless of the reason for the difference, \flit was able to identify this
difference, and if the comparison function can handle different domain sizes,
then it would be able to quantify the variability induced by changing the
parallelism configuration.

Finally, we wanted to verify if the Bisect algorithm can identify the
same files and functions under MPI as it did sequentially.
For this step, we took a single random sample from a successful sequential
Bisect run for each test (except for tests 7, 12, and 19, which had no
successful sequential Bisect runs).
Each random sample was able to isolate the same sets of files and
functions, regardless of the variability introduced by the parallelism.
This approach
may not work all the time
(since the variability produced
by parallelism may cause the code to branch differently).
In the case of \mfem,
this case did not arise;
it is highly encouraging that \flit generates
identical results despite the parallelism.

\section{Related Work}                       
\label{sec:sec-related-work}

\subsection{Reproducibility}

The general areas of floating-point error analysis and result
reproducibility have been receiving a lot of
attention~\cite{SC15-Repro-BOF,balaji2013reproducibility,steyerintel,corden2009consistency,sc16-panel-fp}.
There have also been some efforts in understanding performance and reproducibility
in the setting of GPUs~\cite{nvidia-fp}.
The study of deterministic cross-platform floating point arithmetics
was reported a decade ago in~\cite{christian-seiler} by Seiler.
Our initial work on \flit was inspired by this work.


In \cite{baker-work}, the authors discuss the impact of nonreproducibility
in climate codes.
The tooling they provide (KGEN) is home-grown, not meant for external
use~\cite{paper:kgen}.
Their work does not involve any capability similar to Bisect.
Their focus is on large-scale Fortran support (and
currently \flit does not handle Fortran; it is a straightforward
addition and is future work for us).

A tool called COSFID \cite{paper:cosfid} was used to take climate codes and
analyze them more systematically.
Their work realizes file-level bisection search, albeit
through a single recursive bash script.
Their work does not perform symbol-level bisection to isolate
problems down to individual functions, as we do.


The issue of designing bitwise reproducible applications is
discussed in~\cite{DBLP:conf/ipps/ArteagaFH14}.
Their work focuses on the design of efficient reduction
operators, improving on prior work on deterministic addition.
It does not support capabilities such as  compilations
involving different optimizations, and bisection search.

A recent study has discussed the relative lack
of understanding about floating-point arithmetic amongst
practitioners~\cite{DBLP:conf/ipps/DindaH18}.
The issues we encountered in Laghos (the swap macro that turned out to be
undefined behavior
according to C++, and the non-robust comparison against 0.0)
are both indicative of this observation.
Doug James~et~al. stress the need for widespread education in this
area~\cite{DBLP:journals/corr/JamesWS14}.

\subsection{Performance Tuning}
\label{subsec:related-performance-tuning}

This work implements a very rudimentary performance tuning model of running all
flag combinations within the search space and measure each one.
The novelty here is to allow navigation between performance and
reproducibility.

There is extensive work in the community with more sophisticated performance
tuning techniques.
For example, Profile Guided Optimization (PGO)~\cite{paper:pgo}, also known as Profile-Directed
Feedback (PDF), is implemented in most mainstream
compilers~\cite{%
book:ibm-performance-guide,%
web:msvc-pgo,%
paper:intel-performance,%
paper:gcc-pgo,%
web:clang-manual}.
PGO uses an instrumented compilation to log the places in the code that are
most used and in what order.
This log is then used in a later compilation step to optimize the executable
specifically for that trace.
This approach is useful if you expect your application to follow almost the
same path every time.

Other work has tuned the specific parameters within compilers such as the TACT
tool~\cite{paper:compiler-opt-auto-tuning}.
This tool tunes the internal parameters of the GCC compiler optimizations
for one particular application.
One could take a similar approach with any compiler,
but each would contain its
own internal optimization parameters.

This work primarily focuses on reproducibility and identifying sources of
variability, which at first seems orthogonal to performance.
However, we recognize that one often changes architecture, compiler, or
compiler optimization flags when seeking performance, and it is at these times
that reproducibility can become an issue.
We made an initial attempt to incorporate performance and performance tuning
without detracting from the primary goal of reproducibility.
Involving work from the vast performance tuning community into the \flit work
is left as future work.

\section{Concluding Remarks}                 
\label{sec:conclusion}
The case studies reported in this paper demonstrate that 
porting applications {\em even across today's machines and compilers/flags}
can be quite problematic in the field in terms of result-variability.
For HPC applications developed over decades, 
the problem worsens.
This observation is especially true at
``the end of Moore's law''
where heterogeneity (CPUs, accelerators,
and a plethora of compilers) is the rule and not the exception. 
Our work through \flit
has already impacted state-of-the-art projects at Lawrence Livermore labs, as
we previously described.
Our algorithms have yielded results concerning actual projects, as
well as in the context of fault injection studies
on the LULESH proxy application.

Without tools such as \flit, a programmer
may end up adopting draconian
measures such as prohibiting the project-wide use of optimizations higher
than, say, \texttt{-O2}---something that would be counterproductive.
Tools such as \flit will become increasingly important in supporting
new proposals~\cite{eurollvm-paper-mixing-rounding} for mixing
the use of the fast-math and precise-math modes~\cite{DBLP:books/daglib/0027236} in the same LLVM
compilation.
Such mixings can help relax 
numerical precision in sub-modules where speed matters (and result
variability does not matter as much).
With \flit, one can identify which modules can be
optimized under fast math, thereby supporting the use of these new
LLVM options.

In addition to discovering variability, FLiT can help exercise compiler flag
combinations and discover bugs.
One such bug we discovered during the course of using FLiT involved using
{\tt -Ofast} and {\tt -ffloat-store} has been reported and fixed in
GCC 8.2.0~\cite{ian-gcc-bug-90187}.

We have already begun
applying \flit to popular libraries such as CGAL~\cite{cgal-library}
that find applications in 3D printing and other critical applications.
It was encouraging for us to discover the (relative) ease of integrating
\flit into the building and testing infrastructure of CGAL.
We also have identified specific instances of when it is unsafe to apply higher levels of
optimization, as these can drastically change the computed results (\textit{e.g.}, even discrete
answers such as the number of points on a mesh).
This study also revealed some limitations
of Bisect that we plan to overcome.
As one example,
if an application heavily uses inlining, the granularity
of file bisection search can often reduce to a single file, which is insufficient
for precisely root-causing variability.
Therefore, alternative methods (\textit{e.g.},
dynamic execution based) must be developed.
The community also needs to better address the issue of
communicating the intended levels of optimizations between
developers and users.
Our experience is that without this information, we can
overly optimize an application, only to find it throwing exceptions
or not converging properly.

Going forward, one significant limitation of \flit, namely its inability to
handle application-level non-determinism, must be addressed.
We plan to extend \flit to work under OpenMP, MPI, accelerator/GPU programming, and other forms
of concurrency, with support for result determinization
    provided in an easy-to-use manner.
Where determinization is infeasible, we may have to employ ensemble-based
approaches such as proposed in~\cite{DBLP:journals/corr/abs-1810-13432}.
Last but not least, we will continue to enhance the robustness of \flit.
We continue to maintain the open-source status of \flit, and invite contributions
as well as usage of \flit in others' projects, providing us feedback.

\section{Acknowledgments}
\label{sec:acknowledgments}
This work was
performed under the auspices of the U.S. Department of Energy
by LLNL under contract DE-AC52-07NA27344 (LLNL-CONF-759867),
and supported by NSF CCF 1817073, 1704715.

\bibliographystyle{ACM-Reference-Format}
\bibliography{bisect}

\end{document}